%% file: main.tex
\documentclass[sigconf]{acmart}
\renewcommand\footnotetextcopyrightpermission[1]{} 
\usepackage[ruled,linesnumbered,vlined]{algorithm2e}
\usepackage{enumerate}
\usepackage{subfigure}
\usepackage{subcaption}
\usepackage{amsmath}
\usepackage{amsthm}
\usepackage{fancyhdr}
\begin{document}

\title{Hidden Sketch: A Space-Efficient Reversible Sketch for Tracking Frequent Items in Data Streams}

\author{Zicang Xu}
\affiliation{
\institution{Peking University}
\city{Beijing}
\country{China}
}
\email{xuzicang@stu.pku.edu.cn}
\author{Yuxuan Tian}
\affiliation{
\institution{Peking University}
\city{Beijing}
\country{China}
}
\email{tianyuxuan@stu.pku.edu.cn}
\author{Yuhan Wu}
\affiliation{
\institution{Peking University}
\city{Beijing}
\country{China}
}
\email{yuhan.wu@pku.edu.cn}
\author{Tong Yang}
\affiliation{
\institution{Peking University}
\city{Beijing}
\country{China}
}
\email{yangtong@pku.edu.cn}


\begin{abstract}
  Modern data stream applications demand memory-efficient solutions for accurately tracking frequent items, such as heavy hitters and heavy changers, under strict resource constraints. Traditional sketches face inherent accuracy-memory trade-offs: they either lose precision to reduce memory usage or inflate memory costs to enable high recording capacity. This paper introduces Hidden Sketch, a space-efficient reversible data structure for key and frequency encoding. Our design uniquely combines a Reversible Bloom Filter (RBF) and a Count-Min (CM) Sketch for invertible key and frequency storage, enabling precise reconstruction for both keys and their frequencies with minimal memory. Theoretical analysis establishes Hidden Sketch's space complexity and guaranteed reversibility, while extensive experiments demonstrate its substantial improvements in accuracy and space efficiency in frequent item tracking tasks. By eliminating the trade-off between reversibility and space efficiency, Hidden Sketch provides a scalable foundation for real-time stream analytics in resource-constrained environments.
\end{abstract}





\maketitle
\input{sections/Introduction}
\input{sections/Relatedwork}
\input{sections/Algorithm}
\input{sections/Experiment}
\input{sections/Conclusion}
\bibliographystyle{unsrt}
\balance
\bibliography{reference}
\clearpage	
\appendix
\input{sections/Math}
\end{document}

%% file: sections/Introduction.tex
\section{Introduction}
\label{sec:Intro}
\subsection{Background and Motivation}
In the era of big data, processing massive data streams has become a crucial challenge across various domains.
Mining valuable information in data streams has broad applications such as network traffic monitoring\cite{network}, social network analysis\cite{social}, financial transaction monitoring\cite{financial}, and recommendation systems\cite{recommendation}.
A critical task in data stream analysis is tracking frequent items, such as heavy hitters and heavy changers, which dominate the features of a data stream.
Heavy hitters are the items that frequently occur in a data stream, while heavy changers are the items whose frequency changes heavily between two consecutive time intervals.

In real-world scenarios, the high velocity of incoming items and their unpredictable keys pose significant challenges for recording the desired information. 
To address these challenges, the research community has developed numerous sketch algorithms.
These algorithms are probabilistic in nature and are designed to achieve low memory overhead and high update rates. 
Most sketch algorithms utilize hash table data structures to complete update operations in constant time, often enabling parallel or pipelined execution. 
However, due to the high time cost associated with memory access, these data structures must typically be deployed in on-chip memory, such as L2 cache or shared memory, which are scarce resources on most processors.
Therefore, maintaining high accuracy with less memory overhead has become the development direction of sketch algorithms.

Sketch algorithms are primarily based on hash tables, where each item in the data stream is mapped to multiple buckets in the table using independent hash functions. 
Classic sketch algorithms, such as CM sketch\cite{cmsketch}, CU sketch\cite{cusketch}, and Count sketch\cite{countsketch}, offer both theoretical and practical guarantees for accurate frequency estimation of frequent items.
However, since these traditional sketch algorithms cannot store the item keys, they are limited in applications like anomaly detection.
Tracking keys presents a significant challenge because, while frequency estimates can be computed with reasonable accuracy, recording keys should ensure their completeness and correctness.

In response to this challenge, researchers have developed various algorithms to record item keys.
They can be broadly categorized into two classes based on their key recording strategies.
The first class includes explicit key recording sketch algorithms, which add key fields to each bucket in the hash table to store item key information\cite{ HeavyKeeper, elastic,MVSketch, fcmsketch, TightSketch}.
However, these algorithms are susceptible to hash collisions, necessitating complex strategies for retaining the keys of frequently occurring items.
For example, Elastic Sketch\cite{elastic} employs an Ostracism strategy to evict less frequent keys and retain high-frequency ones.
Nevertheless, these algorithms often require substantial memory to maintain a low miss rate, leading to inefficiencies.

The second class encompasses implicit key recording sketch algorithms, which aim to reduce memory usage by not directly recording keys\cite{Deltoid, revsketch,Sequential,Bitcount,flowradar,Sketchlearn}.
Prior methods have developed key mixing or index-based encoding approaches for implicit key recording.
In the key mixing approach, exemplified by FlowRadar, multiple keys are mixed within the key field.
This method requires auxiliary data structures for decoding, and when the decoding fails, nearly all the recorded information is lost.
In contrast, the index-based encoding approach, such as Reversible Sketch, embeds key information into bucket indices.
While this method avoids auxiliary structures, it demands a large number of buckets to ensure decoding accuracy, leading to potential memory inefficiency.
%

\subsection{Our Solution}
To address the limitations of existing methods, we propose \textbf{Hidden Sketch}, a novel invertible sketch algorithm designed to record keys and frequencies with minimal memory overhead efficiently.
Our algorithm draws inspiration from previous implicit key recording approaches while addressing their key deficiencies.
The core idea behind Hidden Sketch lies in leveraging a hybrid design that integrates a CM Sketch for frequency estimation with a Reversible Bloom Filter for key encoding.
By combining these components, Hidden Sketch fully exploits the information embedded in bucket indices and compactly encodes key information.
Unlike traditional methods, where key encoding may require significant memory, the Reversible Bloom Filter achieves efficiency by using 1-bit buckets, drastically reducing memory requirements.
At the same time, the CM Sketch is modeled as a system of linear equations, enabling accurate decoding of item frequencies when sufficient memory is available.
%

%

%
A key innovation of Hidden Sketch is its robust and systematic decoding process.
The Reversible Bloom Filter identifies candidate keys with high efficiency, avoiding the exhaustive search overhead seen in traditional Bloom Filters.
Meanwhile, we regard the CM Sketch as a system of linear equations, enabling precise decoding of item frequencies through established mathematical methods.
By solving these equations, the algorithm ensures accurate recovery of frequency information for significant items.
This hybrid strategy not only ensures precision but also overcomes the common limitations of prior methods, such as catastrophic information loss in FlowRadar or excessive memory demands in Reversible Sketch.
To track frequent items efficiently, we embed Hidden Sketch into a two-stage framework.
The first stage employs a lightweight filter to pre-process the data stream.
This filter efficiently excludes most low-frequency and unimportant items, ensuring that only significant items are passed to the second stage for precise processing.
This two-stage design ensures that only frequent items are processed in detail, significantly reducing the computational and memory overhead.

The contributions of this paper can be summarized as below:
\begin{itemize}
    \item We propose Hidden Sketch, a novel reversible sketch that can record both the key and frequency of items exactly with high space efficiency and reliability. 
    \item We conduct extensive experiments on different tasks and show that our algorithm outperforms SOTA algorithms.
    \item We prove the memory bound of Hidden Sketch and show the space efficiency theoretically.
    \item We open-source our code on Github for future study\cite{opensource}.
\end{itemize}

The remainder of this paper is organized as follows. 
In \S\ref{sec:Background}, we first formalize the definitions of the data stream and tasks we address in this paper.
Then we conduct a detailed analysis of prior works.
We introduce the core design, Hidden Sketch, in \S\ref{sec:Alg:Design}, a fully invertible data structure that records items implicitly.
Based on the Hidden Sketch, we introduce how to complete tasks such as heavy hitter detection in \S\ref{sec:Alg:Overview}.
%
%
We use experiments to evaluate our solution for heavy hitter detection and heavy change detection tasks in \S\ref{sec:Exp}.
Finally, we conclude this paper in \S\ref{sec:Conclusion}.

%% file: sections/Relatedwork.tex
\section{Background and Related Work}
\label{sec:Background}
In this section, we first formalize the definitions of data stream and the tasks we address in this paper. 
Then, we review existing methods, emphasizing their strengths and limitations.
\subsection{Problem Statement}
\label{sec:Background:problem}
\begin{itemize}
    \item \textbf{Data Stream Model}: A data stream $S$ is formally defined as a sequence of items $\langle e_1, e_2,..., e_{|S|}\rangle(e_i\in \mathcal{K})$, where $|S|$ is the total size of the data stream, and $\mathcal{K}$ is the key space of items. Items with the same key can appear more than once, and the frequency of an item key $e$ is defined as $f(e):=\sum_{e_i=e}1$.
    \item \textbf{Frequency Estimation}: Given a data stream $S$, this task aims to estimate the frequency $f(e)$ of an querying item $e\in \mathcal{K}$. Frequency estimation is the most basic task since many tasks are relative to the frequency of items, including heavy hitter detection and heavy changer detection.
    \item \textbf{Heavy Hitter Detection}: In this task, we want to identify items in a data stream whose frequency exceeds a certain threshold. Formally, given a data stream $S$ and a pre-defined threshold $T$, the task is to find all items $e$ such that $f(e)\geq T$. Note that both the key and the frequency of heavy hitters are needed.
    \item \textbf{Heavy Changer Detection}: This task involves identifying items whose frequency changes significantly between two contiguous windows of a data stream. Formally, give two contiguous and non-overlapping data stream windows $S_1$ and $S_2$, the task is to find items $e\in \mathcal{K}$ such that $|f_{S_1}(e) - f_{S_2}(e)|>\Delta$, where $f_{S_i}(e)$ is the frequency of $e$ in window $S_i$, and $\Delta$ is a pre-defined threshold.
\end{itemize}
Accurate tracking of frequent items forms the core of both heavy hitter and heavy changer detection.
The central challenge lies in achieving high-precision frequency estimation while maintaining the keys of these items.
This dual demand amplifies complexity in resource-constrained environments.

\subsection{Related Work}
\label{sec:Background:RelatedWork}
Existing sketch algorithms for recording keys and frequencies can be broadly classified into two categories: explicit key recording and implicit key recording.
Each category has its unique advantages and limitations, which will be discussed in detail below.

\subsubsection{Explicit key recording:} Explicit key recording sketches\cite{HeavyKeeper, elastic, MVSketch, fcmsketch,  TightSketch,LD-Sketch} add a key field to each bucket in a hash table to directly store the key of the item hashed to the bucket.
However, due to hash collisions, multiple keys may map to the same bucket, and these algorithms can only select one of them to record in the key field.
To address this, various replacement strategies have been proposed to prioritize high-frequency keys.
For instance, HeavyKeeper\cite{HeavyKeeper} employs a count-with-exponential-decay strategy to keep heavy hitters and reduce the effect of low-frequency items.
Similarly, Elastic Sketch\cite{elastic} utilizes an Ostracism strategy to evict unimportant keys from the key fields.
Other algorithms, such as MV-Sketch\cite{MVSketch}, LD-Sketch\cite{LD-Sketch}, and TightSketch\cite{TightSketch}, also employ auxiliary bucket fields and delicate evicting strategies.
The limitation of these algorithms lies in their need to protect the completeness the recorded keys.
In a hash table, some buckets are never mapped by heavy hitters, while others are mapped by multiple heavy hitters due to hash collisions.
The replacement strategies can only solve collisions between high-frequency and low-frequency items, while they are powerless against collisions between high-frequency items.
To guarantee the recall rate of frequent items, these approaches have to allocate enough buckets, making some buckets empty or filled by low-frequency items, which is a waste of memory.
Thus, the space efficiency of these approaches is naturally limited.
\subsubsection{Implicit key recording:} Implicit key recording sketches aim to reduce memory usage by encoding keys into the data structure rather than storing them explicitly\cite{revsketch, Bitcount, flowradar,Sketchlearn, Sequential}.
They brilliantly encode keys into sketch data structures and decode them offline later.
Specifically, there are two main encoding methods at present.
The first encoding method, introduced by FlowRadar \cite{flowradar}, encodes keys by mixing them in a designated key field while maintaining a count of distinct keys in a separate field.
During decoding, the algorithm identifies buckets mapped by a unique key.
The key values in these buckets are extracted iteratively, leading to new buckets mapped by a unique key until either all keys are decoded or all buckets contain multiple distinct keys.
This process can be modeled as a 2-core pruning operation on a hypergraph, where buckets represent nodes and items form hyperedges connecting nodes corresponding to hashed buckets.
Then, the decoding process can be regarded as iteratively finding nodes of 1 degree and then eliminating the nodes and their hyperedge.
Successful decoding is guaranteed if all the hyperedges are eliminated, i.e. no 2-cores exist in the hypergraph.
Although this method only requires a linear number of buckets to achieve a high probability of a successful decoding process, it requires an additional Bloom Filter to identify distinct keys and a field for each bucket to store the number of distinct keys.
Moreover, when the size of the Bloom Filter or the number of buckets is insufficient, the decoding process may fail, and nearly all the key information is lost, as the mixed keys in the key field are not recoverable.
The second method leverages bucket indices in sketches, avoiding additional structures by designing specialized hash functions.
These algorithms exploit the skewed distribution of real-world data streams, where buckets mapped by heavy hitters possess significantly larger values.
By identifying items associated with these buckets, these algorithms effectively decode high-frequency keys.
To the extent of decoding results, these algorithms are equivalent to traversing the whole key space and finding the ones with extremely high frequency.
To avoid the unacceptable overhead of traversal, these algorithms use special hash functions, which can recover the origin key from several hash values.
For instance, Reversible Sketch uses modular hashing, which partitions the entire key into $q$ words, does $q$ hash functions on them respectively, and concatenates the $q$ hash values to get an entire hash function.
When given several bucket indices, Reversible Sketch can easily report possible keys that can be mapped to them by separately handling the partitioned keys.
Although it seems that these algorithms employ no extra memory, their encoding capacity is affected by the number of buckets.
They need enough index space to fully encode the keys, which may exceed the need to report accurate frequency estimation and cause more memory waste.
Moreover, when the threshold is lowered, buckets of heavy hitters cannot be significantly different from others, thus causing inaccurate detection.

\subsection{Summary}
Our analysis reveals that existing sketch designs face a fundamental dilemma between space efficiency and tracking reliability.
The explicit key recording suffers from inherent space inefficiency due to hash collision, and implicit approaches struggle with decoding fragility.
We propose \textbf{Hidden Sketch}, a novel implicit approach that independently encodes keys and frequencies, supporting a reliable decoding process and high space efficiency.

%% file: sections/Algorithm.tex
\section{Hidden Sketch Design}
\label{sec:Alg:Design}
We adopt the index-based encoding method to record keys.
Prior approaches using this method directly use the counter index to encode keys.
This causes redundant counters for frequency recording.
Therefore, we separate the key encoding part from the frequency recording part.
In the key encoding part, we use the index of 1-bit buckets to encode the keys.
On this occasion, the CM Sketch degenerates into a Bloom Filter.
We apply the index-based encoding method to the Bloom Filter and propose the Reversible Bloom Filter(RBF).
The reversibility of RBF is reflected in that it provides a feasible key recovery process while not reducing the false positive rate of Bloom Filters.

As for the frequency recording part, we utilize a CM Sketch to enable a systematic decoding process.
Specifically, after we decode the existing keys from the RBF, we can decode their exact frequencies by solving the linear equation system established by the CM Sketch.
Meanwhile, the false positives reported by the RBF can be filtered since their decoded frequencies are zero.

\subsection{Reversible Bloom Filter(RBF)}\label{sec:Alg:Design:RBF}
\subsubsection{Data Structure:}
\begin{figure*}
    \centering
    \includegraphics[width=0.9\textwidth, ]{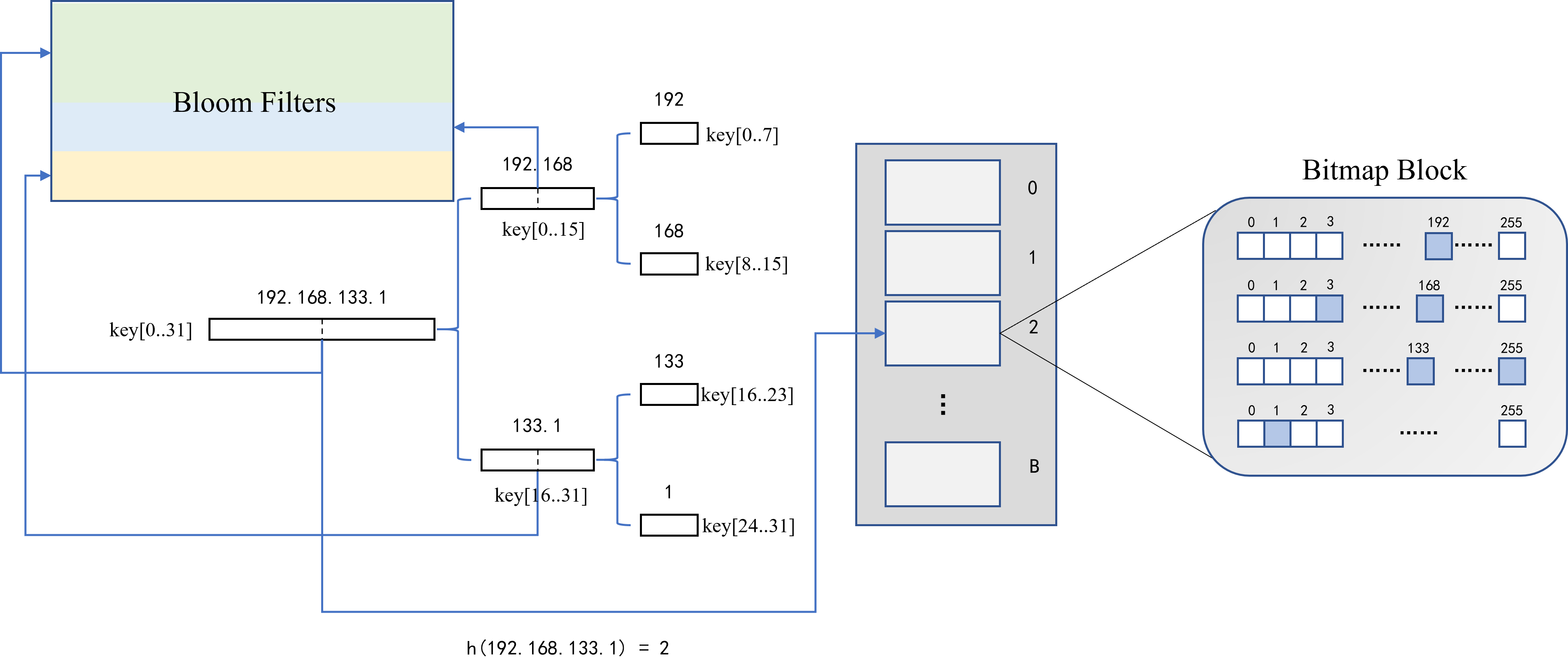}
    %
    \caption{An example of Reversible Bloom Filter with 32-bit key.}
    \label{Fig:RBF}
\end{figure*}
The Reversible Bloom Filter(RBF) employs a hierarchical structure to encode keys efficiently.
On the one hand, the RBF is essentially a Bloom Filter with special hash functions for index-based key encoding.
On the other hand, we divide the key conceptually based on a hierarchical tree structure and allocate a part of the memory for each tree node.

We illustrate a Reversible Bloom Filter recording 32-bit keys as an example in Fig. \ref{Fig:RBF}.
We divide the 32-bit key into four 8-bit base key segments corresponding to the tree's leaf nodes.
The internal nodes represent concatenations of segments from their child nodes, progressively combining segments as one moves up the tree.
We allocate a Bloom Filter for each internal node, represented by different colors in the block of Bloom Filters in Fig. \ref{Fig:RBF}.
For the leaf nodes, we use a bitmap block array instead of Bloom Filters due to the small segment spaces.
Each block contains multiple bitmaps that correspond to the leaf nodes respectively.
When a key is inserted, we obtain the segments of the internal nodes and add them into the corresponding Bloom Filters.
Simultaneously, we hash the complete key to a block in the bitmap block array and insert its leaf segments into their corresponding bitmaps in the hashed block.
We also illustrate an insertion example of RBF in Fig. \ref{Fig:RBF}.
When inserting the key 192.168.133.1, we add its internal node segments, including 192.168, 133.1, and 192.168.133.1, to the corresponding Bloom Filters.
We also use a hash function $h$ that maps 192.168.133.1 to the second block in the bitmap block array.
Within the block, there are four bitmaps corresponding to the four leaf nodes.
We set the 192nd bit, 168th bit, 133rd, and 1st bits of the four bitmaps to 1, respectively.

The Reversible Bloom Filter is also a Bloom Filter but with delicately designed hash functions.
Specifically, the Bloom Filters of internal nodes can be seen as employing hash functions that only depend on the corresponding key segment.
And the bitmap block array can be viewed as a Bloom Filter that uses hash functions $Hash(key) = h(key) * 2^{l} + seg(key)$, where $h(\cdot)$ is the hash function that maps $key$ to a certain block, $l$ is the length of leaf node segment, and $seg(\cdot)$ is the function that maps the key to its corresponding segment.
\begin{algorithm}
    \caption{Key Set Recovering Process}
    \label{pseudo:KeySetRecovering}
    \SetKwFunction{getkeys}{getCandKeys}
    let $root$ be the root node representing the entire key\;
    $S=\varnothing$\;
    \For {$i$ \rm{in [1..$B$]}} {
        $S= S\cup \sigma_{h(key)=i}(\getkeys{root, i})$\;
    }
    \KwRet{$S$}\;
    \SetKwProg{Fn}{Def}{:}{}
    \Fn{\getkeys{$node, index$}}
    {
        \If {node \rm{is leaf}}{
            $S=\varnothing$\;
            let $M$ be the bitmap corresponding to $node$ in $bitmapBlock[index]$\;
            \For {i \rm{in [0..$2^{node.l}$]} } {
                \If{\rm{$M$[$i$]==1}}{
                    $S=S\cup \{i\}$\;
                }
            }
            \KwRet{S}\;
        }
        \Else{
            $S=\{\epsilon\}$\;
            \For {\rm{each $child$ of $node$}} {
                $S=S\times\getkeys{child, index}$\;
            }
            $S=\sigma_{BF_{node}}(S)$\;
            \KwRet{S}\;
        }
    }
\end{algorithm}

\subsubsection{Key Set Recovery:}
The key recovery process in the Reversible Bloom Filter is efficient and systematic, enabling the reconstruction of encoded keys through a bottom-up approach in the tree.
For each bitmap block in the array, we obtain a candidate set of each leaf segment by examining the bitmaps within the block.
Next, we recursively concatenate the candidate sets of segments to create longer segment candidate sets, ultimately resulting in the candidate set of the complete key.
Specifically, for each internal node $\mathcal{N}$, once we yield all the candidate sets of its child nodes, we can compute the Cartesian product of these sets to form a candidate key segment set of $\mathcal{N}$.
Then we filter the candidate set using the Bloom Filter of $\mathcal{N}$.
Formally, if an internal node $\mathcal{N}$ has $k$ child nodes with candidate sets $S_1,S_2,...,S_k$, the candidate set of $\mathcal{N}$ is computed as follows:
\begin{equation*}
    S=\sigma_{BF_\mathcal{N}}\{seg_1\oplus ...\oplus seg_k|seg_i\in S_i, i=1,2,...,k\}
\end{equation*}
In this equation, $\oplus$ represents the concatenation operator for binary strings, and $\sigma_{BF_\mathcal{N}}$ denotes the filtering operation based on the Bloom Filter associated with $\mathcal{N}$.
After computing the candidate sets for all internal nodes, the root node generates the candidate set for the complete key.
Each candidate key is then verified against the bitmap block array to ensure it hashes to the correct block.
Algorithm~\ref{pseudo:KeySetRecovering} illustrates the pseudocode of the recovering process.
We utilize a recursive function \texttt{getCandKeys} to compute the candidate key set corresponding to $node$ induced by the bitmap block $index$.
If $node$ is a leaf key segment node, it uses the corresponding bitmap in block $index$ to recover the candidate key set (lines 7-13).
If $node$ has child nodes, it computes the Cartesian product of the candidate key set from its child nodes and then filters this set using the corresponding Bloom Filter (lines 16-18).
The recovery process computes the union of the candidate key sets induced by the $B$ blocks and filters the candidate keys using the block hash function $h$ (lines 2-4).

It is important to note that the output set of the recovery process consists entirely of the keys whose mapping positions are set to 1 in the Reversible Bloom Filter. In other words, the recovery process is functionally equivalent to a traversal method that checks and verifies the entire key space within the Reversible Bloom Filter in terms of results. However, its time complexity is more efficient compared to the traversal method. 
This leads us to prove that the false positive rate of the Reversible Bloom Filter is comparable to that of a standard Bloom Filter. We will discuss this in detail in Appendix \ref{sec:Math}.

\subsection{Frequency Decoding with CM Sketch}\label{sec:Alg:Design:DecodingProcess}
We select the CM Sketch for frequency recording since it can be regarded as a linear mapping\cite{SynopsesForMassiveData, NZE}.
Formally, suppose the number of possible items is $n$, and the number of buckets employed by the CM Sketch is $m$.
Let $\vec{x}$ represents a column vector with $n$ dimensions, where the $i$th element corresponds to the frequency of item $i$.
Let $\vec{y}$ denotes a column vector with $m$ dimensions, where the element on the $j$th dimension presents the value of bucket $j$.
Then the CM Sketch can be mathematically expressed by the following equation:
\begin{equation*}
    \mathbf{\Phi}\cdot\vec{x}=\vec{y},
\end{equation*}
where $\Phi$ is a $m\times n$ matrix.
The element $\mathbf{\Phi}_{i, j}$ indicates that the $i$th bucket is incremented by $\mathbf{\Phi}_{i,j}$ if inserting an item $i$.
As we have yielded the candidate key set from the key recovery process of RBF, we can construct the matrix $\mathbf{\Phi}$ that maps the frequency vector of candidate keys to the value vector of buckets.
$\vec{y}$ can be directly obtained from the recorded values in the CM Sketch.
Therefore, we can decode the frequencies of candidate keys by solving the equation $\mathbf{\Phi}\cdot\vec{x}=\vec{y}$.
However, the unique solution of the equation exists only when the rank of $\mathbf{\Phi}$ equals the dimension of $\vec{x}$.
The frequencies can not be determined when the null space of $\mathbf{\Phi}$ is not empty.
Since the construction of matrix $\mathbf{\Phi}$ depends on the hash functions of the CM Sketch, it can be regarded as a random matrix with the same sum on each column.
Intuitively, when the matrix has sufficient rows, it is likely to be full rank.
To illustrate the upper bound of the sufficient number, we introduce the pure bucket extraction process as the first step of our frequency decoding process.
\begin{algorithm}
    \caption{Pure bucket extraction}
    \label{pseudo:PureBucketExtraction}
    \KwIn{The candidate key set $S$, the vector of bucket value $Bucket[m]$.}
    \KwOut{A map of key-frequency pairs.}
    Construct a set array $keySet[m]$, where $keySet[i]$ contains the keys mapping to $Bucket[i]$.\;
    $pureQue\gets $ empty queue\;
    \For{i \rm{in}$[1..m]$}
    {
        \If{$keySet[i]$.\rm{size}$()=1$}{
            $pureQue$.enqueue($i$)\;
        }
    }
    $Result\gets$ empty map\;
    \tcc{iterative resolution}
    \While{\rm{not} $pureQue$.\rm{isEmpty()}}
    {
        $i\gets pureQue$.dequeque()\;
        \If{keySet\rm{$[i]$.isEmpty()}}
        {
            \textbf{continue}\tcp{the key has been removed}
        }
        let $key$ be the unique element in $keySet[i]$\;
        $freq\gets Bucket[i]$\;
        $Result$.insert$(key, freq)$\;
        \tcc{extract the item from the buckets it maps}
        \For{\rm{bucket index $k$ hashed by $key$}}
        {
            $Bucket[k]\gets Bucket[k]-freq$\;
            $keySet[k]$.remove$[key]$\;
            \If{keySet[k].\rm{size}$()=1$}
            {
                $pureQue$.enqueue($k$)\;
            }
        }
    }
   
    \KwRet{Result}\;
\end{algorithm}

The pure bucket extraction step is inspired by the decoding process of FlowRadar\cite{flowradar}.
We identify a bucket of the CM Sketch as a pure bucket if only one candidate key is mapped to the bucket.
In other words, the corresponding row in the matrix $\mathbf{\Phi}$ has only one non-zero element.
We can directly determine that key's frequency as the bucket's value.
After determining a key's frequency, we can extract it from other buckets it maps to and update their values.
The extraction operation may cause the updated bucket to become a new pure bucket.
The pure bucket extraction step is illustrated in Algorithm \ref{pseudo:PureBucketExtraction}.
We maintain a queue for the pure bucket, which is initialized by scanning all the buckets and enqueue the pure ones (lines 2-5).
Then we iteratively extract items and find new pure buckets (lines 7-18).
The loop is terminated when the queue is empty.
That is to say, there are no buckets containing only one key.
If all the buckets are empty, the frequency decoding process is completed, i.e., all the item frequencies are decoded.
As it has been proved in \cite{matrixrank,iblt}, when using $k$ hash functions, the probability of failing to decode all the items is bounded by $O(n^{-k+2})$, if $m>c_kn$, where $c_k$ is a constant associates with $k$.
%

%

The second step of the frequency decoding process addresses scenarios where the pure bucket extraction step fails to decode all the items.
This failure occurs when every remaining bucket contains at least two keys, making it impossible to directly determine the frequency of any key through pure bucket analysis.
In such cases, we resort to solving the equation $\mathbf{\Phi}\cdot\vec{x}=\vec{y}$ based on Singular Value Decomposition(SVD).
SVD is a robust method for solving linear equations, particularly when the coefficient matrix $\mathbf{\Phi}$ is not of full rank.
It decomposes $\mathbf{\Phi}$ into three matrices $U\Sigma V^T$, where $U, V$ are orthogonal matrices, and $\Sigma$ is a diagonal matrix containing the singular values of $\mathbf{\Phi}$.
After that, we can compute a pseudo-inverse of $\mathbf{\Phi}$:
\begin{equation*}
    \Phi^+=V\cdot \Sigma^+\cdot U^T,
\end{equation*}
where $\Sigma^+$ is also a diagonal matrix which transforms the non-zero elements of $\Sigma$ to their reciprocals.
Then we can obtain a possible solution $\vec{x}=\Phi^+\cdot\vec{y}$.
When the rank of $\Phi$ equals the number of its columns, such a solution is the unique solution $\vec{x}$.
Otherwise, it is a possible solution with the least L2-norm.

Although we can directly derive the solution without the pure bucket extraction step, the SVD-based decoding method has some key limitations.
The computation of SVD has a time complexity of $O(n^3)$, which can be very slow for large matrices.
Meanwhile, the computation of SVD and pseudo-inverse matrices suffer from numerical precision issues, leading to inaccurate results.
By integrating the pure bucket extraction with the SVD-based decoding, we achieve a comprehensive frequency decoding framework.
The pure bucket extraction step efficiently resolves buckets with unique keys, reducing the complexity of the problem.
While the SVD-base decoding step handles the remaining unresolved buckets leveraging mathematical guarantees for approximate or exact solutions.
This two-step process ensures high accuracy and scalability, making it well-suited for large-scale data streams.

\subsection{Optimization}
To improve the efficiency of the frequency decoding process, we introduce an optional optimization for the CM Sketch, leveraging the fact that only non-negative integer solutions are required.
In the traditional CM Sketch insertion method, an item's frequency is incremented by 1 in the hashed buckets.
To enhance the decoding process, we propose replacing this increment with a prime number derived from the item's hash.
Specifically, we maintain an array of large prime numbers, denoted as PRIME[].
For each incoming item $e$, a prime number $p_e$ is associated based on the hash function $g(\cdot)$, specifically $p_e = $PRIME$[g(e)]$. The corresponding buckets are then incremented by $p_e$.

The modification on the CM sketch changes the non-zero elements of the equivalent matrix $\Phi$ into random primes instead of 1.
This makes the integer solution sparser in the solution space.
Formally, the matrix of the modified CM sketch can be written as:
\begin{equation}
\label{eq:primeeq}
    \Phi\cdot\Sigma_p\cdot\vec{x}=\vec{y},
\end{equation}
where $\Phi$ is also the matrix constructed by the mapping relationship between items and buckets, and $\Sigma_p$ is a diagonal matrix that multiplies the elements of $\vec{x}$ with primes.
Suppose $\vec{x}_{a}$ be the actual vector of item frequencies, then the solutions of Equation (\ref{eq:primeeq}) can be represented by $\hat{\vec{x}}=\vec{x}_a+\vec{x}_\epsilon$, where $\vec{x}_\epsilon$ is in the null space of $\Phi\cdot \Sigma_p$.
Note that $\Sigma_p\cdot\vec{x}_\epsilon$ can be also seen as a vector in the null space of $\Phi$.
However, its elements are divisible by the corresponding prime number, which makes it sparse in the solution space.

To solve the equation of modified CM Sketch, we also do the pure bucket extraction step to minimize the number of undetermined items.
However, when we find there are multiple solutions in the SVD step (the rank of matrix $\Phi$), we use Integer Linear Programming (ILP) to solve the remaining items.
The constraint set of ILP comprises the unsolved equations and the non-negative constraint on each variable, while the optimization object is not considered.
Using this method, the decoding process can yield actual frequencies even when the matrix is not full rank.
Thus, the encoding capacity of Hidden Sketch is enhanced.

Since we often use 4-byte counters for frequent item recording, surpassing the maximum counter significantly, some bits of the counters are never used in traditional CM Sketch.
This optimization cleverly exploits the unused bits and enhances the encoding capacity of a CM Sketch.

\section{Implementation}
\label{sec:Alg:Overview}
In this section, we illustrate how to use Hidden Sketch to track frequent items and solve the heavy hitter detection and heavy change detection tasks.
Note that heavy changers are also heavy hitters in at least one window.
We can detect heavy hitters in both the two consecutive windows and find items whose reported frequencies change heavily between the two windows.

As the number of heavy hitters is relatively small compared with the total items, we do not want to precisely record insignificant items.
Since most items in a data stream are infrequent, we employ a two-stage framework that separates frequent items from infrequent items.
For frequent items, we record their keys and frequencies accurately, while for infrequent items, we only provide frequency estimation to save memory.
As depicted in Fig. \ref{Fig:Framework}, the first stage employs a lightweight cold filter to pre-process incoming items.
The filter estimates item frequencies in real time and excludes most low-frequency items.
Only items exceeding a predefined filtering threshold are passed to the second stage for precise recording.
For instance, if the threshold of heavy hitter is predefined as 200, we can use a CU sketch\cite{cusketch} with 8-bit buckets width as the code filter.
Each incoming item is hashed into multiple buckets, and the bucket with the lowest value is incremented.
If the value after the increment exceeds 200, the item is also inserted into the Hidden Sketch in the second stage; otherwise, no further action is taken.
This design ensures that the second stage focuses solely on high-value data, reducing both computational and memory overhead.
%

%

When handling heavy changer detection, we obtain the heavy hitters in two consecutive windows.
For each heavy hitter, we query its frequency in the other window.
Specifically, if it is also a heavy hitter in the other window, the reported frequency is the decoded result after adding the filtering threshold.
Otherwise, we query it in the cold filter using its query operation.
Then, we check whether the absolute value of the difference between the item frequencies in the two windows exceeds the heavy change threshold.
If so, we report the item and its change value.

\begin{figure}[t]
\centering
\includegraphics[width=0.5\textwidth]{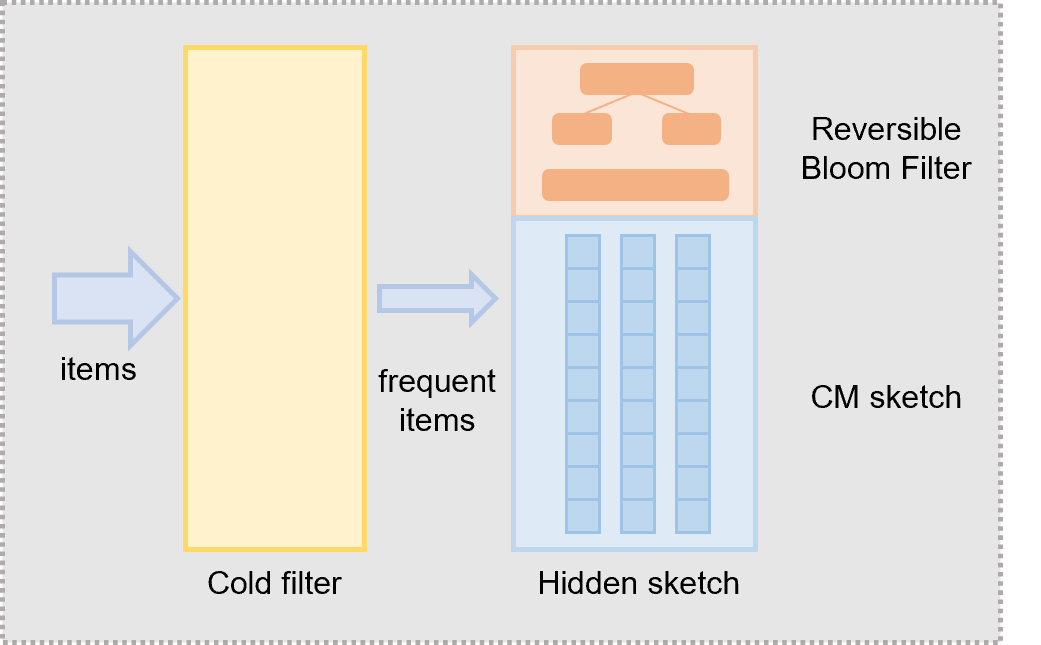} 
\caption{Two-stage framework for tracking frequent items.}
\label{Fig:Framework}
\end{figure}

%% file: sections/Experiment.tex
\section{Experiment}
\label{sec:Exp}

In this section, we use experiment results to evaluate the performance of our algorithm on the tasks given in \ref{sec:Background:problem}.
The experiment results show that our algorithm outperforms previous approaches on frequent item tracking tasks.
It can achieve relatively high accuracy when the memory is extremely limited.

\subsection{Experimental Setup}
\subsubsection{Datasets}
We use three real-world datasets to generate our workloads. 

\begin{itemize}
    \item \textbf{CAIDA:} We use public traffic traces of CAIDA\cite{caida}, which records 5-tuple of each packet in a data center. We divide each trace into 5s-long time intervals, which contain about 2.2M items and 60K distinct keys in each interval.
    \item \textbf{MAWI:} The MAWI dataset is sourced from a traffic data repository maintained by the MAWI Working Group of the WIDE Project\cite{mawi}. Each trace is divided into 45s-long time intervals containing approximately 2.5M items and 50K distinct keys.
    \item \textbf{IMC:} The IMC dataset comes from an empirical study of the network-level traffic characteristics of current data centers\cite{benson2010network}. We set time intervals so that there are around 2M items and 9K distinct keys within each interval.
\end{itemize}

\noindent
\subsubsection{Tasks}
\begin{itemize}
    \item \textbf{Frequency estimation: } The frequency estimation task queries all the items in a time window and reports their frequency estimation. It reflects the most basic feature of data streams.
    \item \textbf{Heavy hitter detection: } We use algorithms to detect heavy hitters whose frequency exceeds 0.01\% of the total frequency. Algorithms should provide the set of detected heavy hitters and report their accurate frequencies simultaneously.
    \item \textbf{Heavy changer detection: } Heavy changers are items whose frequency changes heavily in two consecutive time windows. We use algorithms to detect heavy changers whose frequency change exceeds 0.05\% of the total change.
\end{itemize}

\subsubsection{Evaluation metrics.}
\begin{itemize}
\item \textbf{F1 score: } $\frac{2\times PR \times RR}{PR + RR}$, where $PR$ denotes the precision rate, and $RR$ denotes the recall rate. 
We use F1 score to evaluate the accuracy of heavy hitter and heavy changer detection.
\item \textbf{ARE(Average Relative Error): } $\frac{1}{n}\Sigma_{i=1}^{n}\frac{|f_i-\hat{f_i}|}{f_i}$, where $n$ is the number of distinct items, $f_i$ and $\hat{f_i}$ are the true and estimated frequency of item $i$ respectively. 
We use ARE to evaluate the precision of frequency estimation of heavy hitters and all the items.
%
%
\end{itemize}
\begin{figure*}[htbp]
    \centering
    \subfigure[ARE of frequency estimation]{
        \includegraphics[width=0.23\textwidth,]{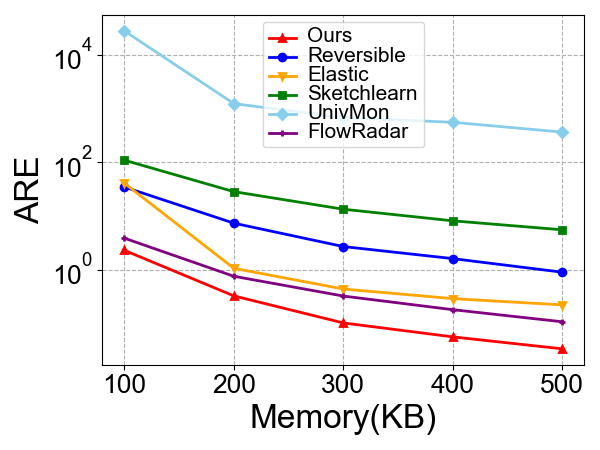}
        \label{Fig:caida:estARE}
    }
    \hspace{-0.2cm}
    \subfigure[F1 score of heavy hitter detection]{
        \includegraphics[width=0.23\textwidth,]{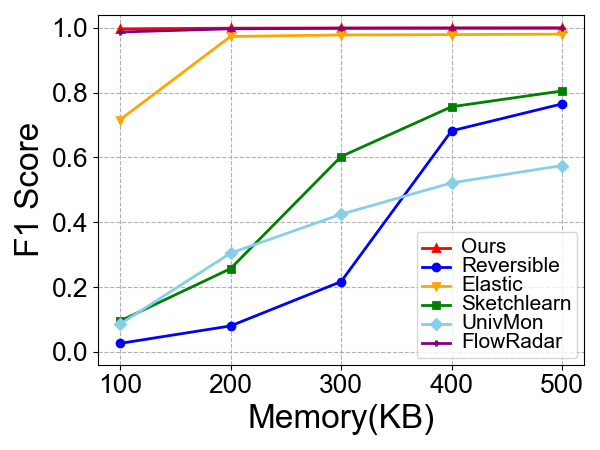}
        \label{Fig:caida:HHF1}
    }
    \hspace{-0.2cm}
    \subfigure[ARE of heavy hitter detection]{
        \includegraphics[width=0.23\textwidth,]{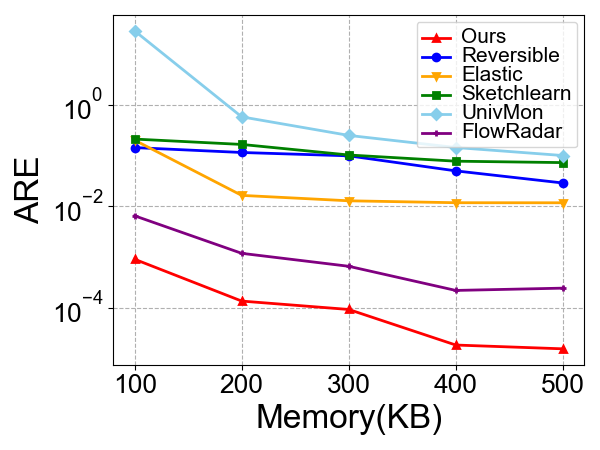}
        \label{Fig:caida:HHARE}
    }
    \hspace{-0.2cm}
    \subfigure[F1 score of heavy changer detection]{
        \includegraphics[width=0.23\textwidth,]{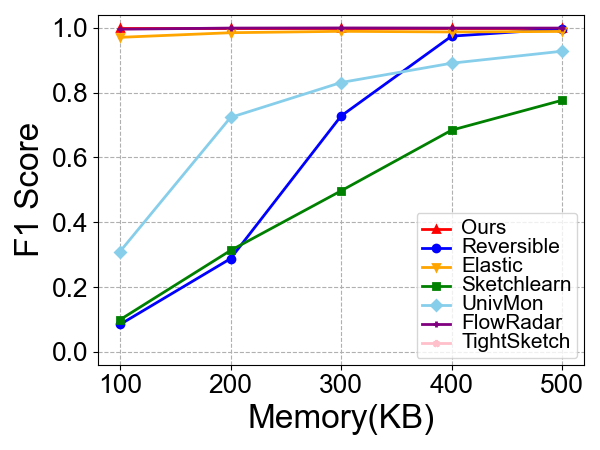}
        \label{Fig:caida:HCF1}
    }
    \caption{Accuracy comparison of our algorithm with baselines on CAIDA.}
    \label{Fig:caida:Accuracy}
\end{figure*}
\begin{figure*}[htbp]
    \centering
    \subfigure[ARE of frequency estimation]{
        \includegraphics[width=0.23\textwidth,]{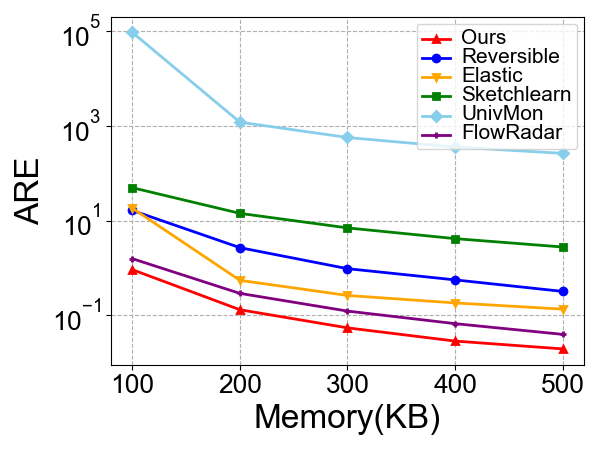}
        \label{Fig:mawi:estARE}
    }
    \hspace{-0.2cm}
    \subfigure[F1 score of heavy hitter detection]{
        \includegraphics[width=0.23\textwidth,]{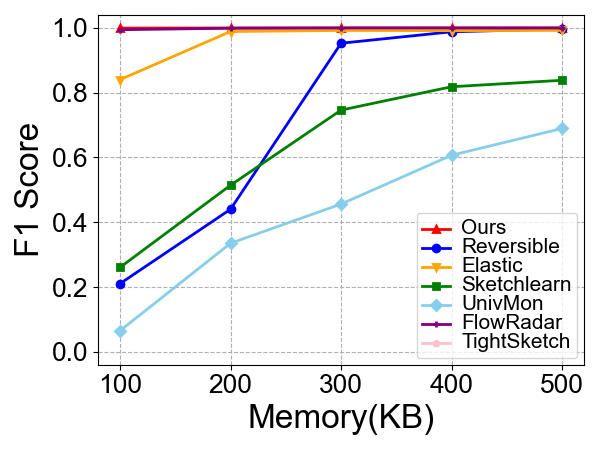}
        \label{Fig:mawi:HHF1}
    }
    \hspace{-0.2cm}
    \subfigure[ARE of heavy hitter detection]{
        \includegraphics[width=0.23\textwidth,]{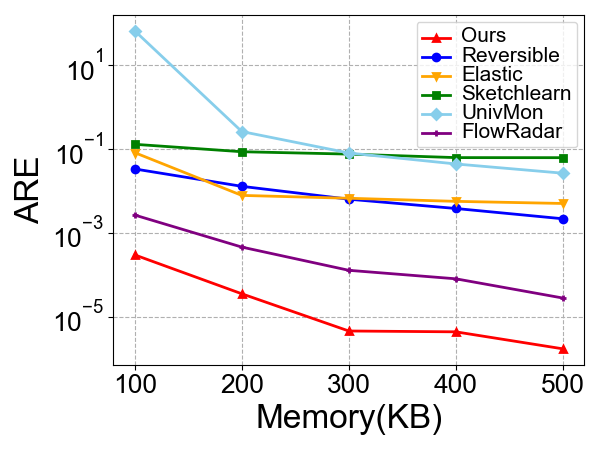}
        \label{Fig:mawi:HHARE}
    }
    \hspace{-0.2cm}
    \subfigure[F1 score of heavy changer detection]{
        \includegraphics[width=0.23\textwidth,]{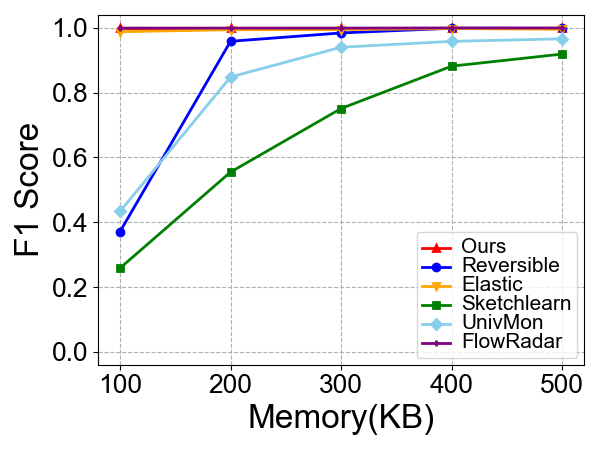}
        \label{Fig:mawi:HCF1}
    }
    \caption{Accuracy comparison of our algorithm with baselines on MAWI.}
    \label{Fig:mawi:Accuracy}
\end{figure*}
\begin{figure*}[htbp]
    \centering
    \subfigure[ARE of frequency estimation]{
        \includegraphics[width=0.23\textwidth,]{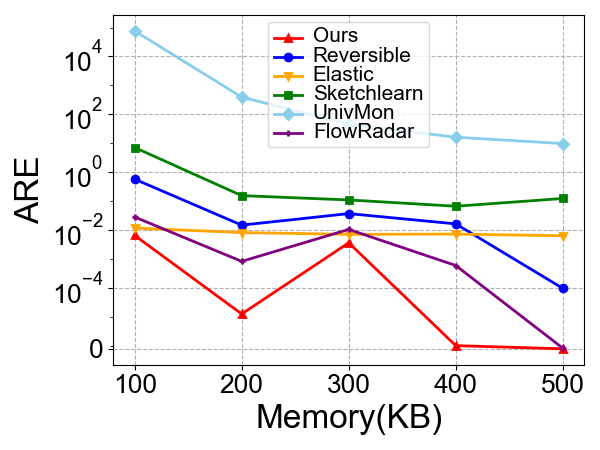}
        \label{Fig:imc:estARE}
    }
    \hspace{-0.2cm}
    \subfigure[F1 score of heavy hitter detection]{
        \includegraphics[width=0.23\textwidth,]{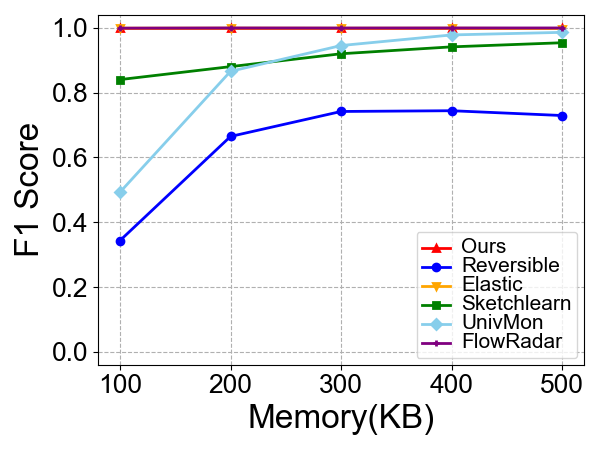}
        \label{Fig:imc:HHF1}
    }
    \hspace{-0.2cm}
    \subfigure[ARE of heavy hitter detection]{
        \includegraphics[width=0.23\textwidth,]{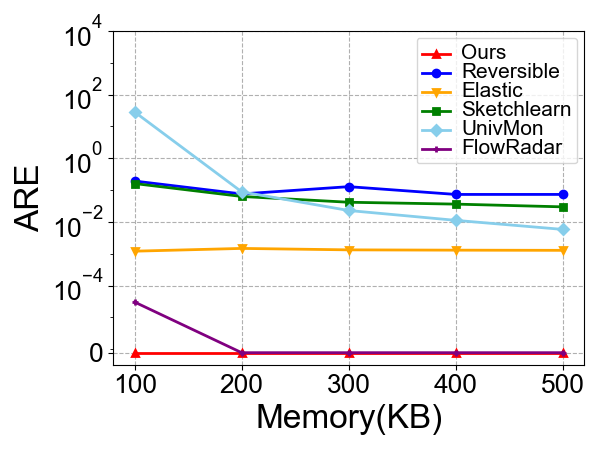}
        \label{Fig:imc:HHARE}
    }
    \hspace{-0.2cm}
    \subfigure[F1 score of heavy changer detection]{
        \includegraphics[width=0.23\textwidth,]{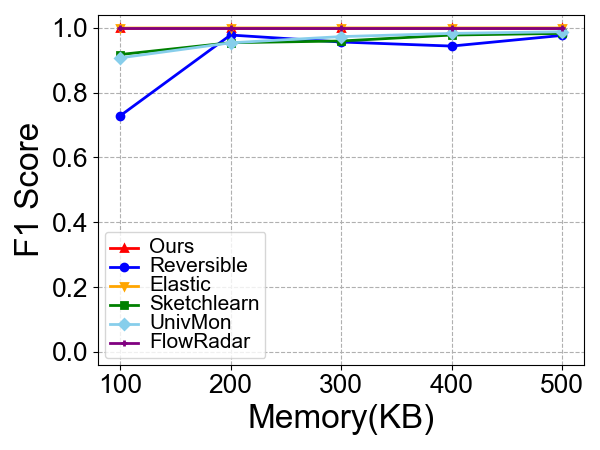}
        \label{Fig:imc:HCF1}
    }
    \caption{Accuracy comparison of our algorithm with baselines on IMC.}
    \label{Fig:imc:Accuracy}
\end{figure*}
\subsubsection{Algorithms and parameters.}
We compare our algorithm with several algorithms, including FlowRadar\cite{flowradar}, Reversible Sketch\cite{revsketch},  Sketchlearn\cite{Sketchlearn}, Elastic Sketch\cite{elastic}, and UnivMon\cite{univmon}. 
Reversible Sketch and Sketchlearn are implicit approaches that can detect heavy hitters and heavy changers, while Elastic Sketch and UnivMon are explicit approaches. 
We select these algorithms as baselines since they can perform well on all the tasks we focus on.
%
\begin{itemize}
    \item\textbf{FlowRadar:}
    Although FlowRadar records all the item keys and frequencies, we can embed it into the two-stage framework.
    We allocate enough memory for the FlowRadar to keep 1800 distinct items and use the remaining memory for the cold filter.
    \item \textbf{Reversible Sketch: }
    For Reversible Sketch, we use 4 arrays as recommended in \cite{revsketch}.
    We adjust the width of arrays to adapt to different memory.
    We allocate the spare memory to the second phase, which is a CM sketch using 3 hash functions.
    \item \textbf{Elastic Sketch: }
    We use the hardware version of Elastic Sketch.
    We allocate 2048 buckets for each array.
    \item \textbf{Sketchlearn: }
    We use one bucket array in each level, which is the best parameter in our experiment result.
    \item \textbf{UnivMon: }
    We use 14 levels for UnivMon, where each level contains a Count sketch\cite{countsketch} with 5 arrays.
    For each level, we use a heap of 900 buckets to record the heavy hitters.
    \item \textbf{Ours: }
    We divide the 32-bit key into 4 8-bit partial keys and organize them in the tree structure as illustrated in Figure \ref{Fig:RBF}.
    We reserve enough memory for the Hidden Sketch with a capacity of 1800 items, which is about 20KB of memory.
    For the cold filter, we select a CU sketch with 8-bit counters.
\end{itemize}
\subsection{Results and Analysis}
Figures \ref{Fig:caida:Accuracy}, \ref{Fig:mawi:Accuracy}, and \ref{Fig:imc:Accuracy} illustrate the accuracy of the algorithms on the three datasets.
We test the algorithms under memory budgets ranging from 100KB to 500KB, demonstrating that our method consistently achieves the highest accuracy on all tracking tasks and datasets.
\subsubsection{Frequency Estimation}
Figures \ref{Fig:caida:estARE}, \ref{Fig:mawi:estARE}, and \ref{Fig:mawi:estARE} illustrate the ARE of frequency estimation of different algorithms.
We find that Hidden Sketch consistently outperforms all baselines.
On the CAIDA dataset, Hidden Sketch achieves an ARE reduction of at least 68.1\% compared to the closest competitor.
On the MAWI dataset, the ARE of Hidden Sketch is at least 41.1\% less than baseline algorithms.
As for the IMC dataset, Hidden Sketch also achieves the lowest ARE on each memory overhead. When the memory is 500KB, the ARE even reduces to 0.

The superior ARE performance of Hidden Sketch in frequency estimation can be attributed to its efficient memory allocation strategy.
The majority of the ARE comes from estimation errors of infrequent items.
Hidden Sketch dedicates a small portion of memory to accurately record frequent items, while the remaining memory is used for frequency estimation with smaller-sized counters.
This design significantly increases the number of available counters, thereby improving estimation accuracy.
\subsubsection{Heavy Hitter Detection}
Figures \ref{Fig:caida:HHF1}, \ref{Fig:mawi:HHF1}, and \ref{Fig:imc:HHF1} present the F1 score of different algorithms on the heavy hitter detection task.
On the three datasets, the F1 score of Hidden Sketch is always near to 1, even when the memory is limited to 100KB.
Implicit approaches such as Reversible Sketch and Sketchlearn employ unreliable recovering processes, while explicit approaches such as Elastic Sketch require enough buckets to reserve heavy hitters.
Therefore, they could not achieve comparable F1 scores on heavy hitter detection tasks when the memory is limited.

Figures \ref{Fig:caida:HHARE}, \ref{Fig:mawi:HHARE}, and \ref{Fig:imc:HHARE} represent the ARE of different algorithms for heavy hitters' frequency estimation.
The ARE of Hidden Sketch is always 1 order lower than the best baseline algorithms.
Especially on the IMC dataset, Hidden Sketch provides zero-error estimation for heavy hitters.
Different from baseline algorithms, the decoding process of Hidden Sketch recovers the exact frequency of items.
The error of frequent items only comes from its online estimation in the cold filter before it is inserted into the Hidden Sketch.

\subsubsection{Heavy Changer Detection}
Figures \ref{Fig:caida:HCF1}, \ref{Fig:mawi:HCF1}, and \ref{Fig:imc:HCF1} show the F1 score of different algorithms on the heavy changer detection task.
Hidden Sketch also achieves F1 scores near to 1 in all the datasets and outperforms other algorithms.
Hidden Sketch can report frequent item sets as the candidate heavy changers, and it can estimate item frequencies accurately.
Therefore, Hidden Sketch can report heavy changers accurately by monitoring the frequency changes of frequent items.
Although Reversible Sketch and Sketchlearn can detect heavy changers through the difference in the sketches of two windows, their recovery processes are unreliable since items with similar frequencies would confuse each other.

%% file: sections/Conclusion.tex
\section{Conclusion}
\label{sec:Conclusion}
In this paper, we propose Hidden Sketch, a novel invertible data structure designed to efficiently record both the keys and frequencies of significant items in high-speed data streams with minimal memory overhead.
By integrating a CM Sketch for frequency estimation and a Reversible Bloom Filter for implicit key encoding, Hidden Sketch achieves superior memory efficiency and accuracy. 
%
%
To address tracking frequent item tasks, we introduce a two-stage framework that effectively minimizes memory usage by filtering out insignificant items.
We conduct extensive experiments to evaluate the memory efficiency and accuracy of Hidden Sketch. The results demonstrate that it consistently outperforms state-of-the-art sketches in item recording tasks, making it a robust and efficient solution for high-speed data stream processing.

\section{Acknowledgments}
This work was supported by Beijing Natural Science Foundation (Grant No. QY23043). The authors gratefully acknowledge the foundation's financial assistance, which enabled access to critical research resources and facilitated the completion of this project.

%% file: sections/Math.tex
\section{Mathematical Analysis} 
\label{sec:Math}
We prove the memory efficiency of Hidden Sketch theoretically in this section.
The analysis of the memory overhead can be divided into two parts.
For the Reversible Bloom Filter, we analyze the memory overhead to achieve an accurate candidate key set recovering.
For the CM Sketch, we will show that with a linear number of buckets relative to the number of keys, the established equation system have a unique solution, which means that we can determine the precise frequency of the items in the Hidden Sketch.

Since our Reversible Bloom Filter is a Bloom Filter with delicately designed hash functions, we need to review the memory overhead of Bloom Filter\cite{BFTheorem}.
After that, we will show that such hash function design will not reduce the performance when choosing appropriate parameters.
\begin{lemma}
\label{lemma:BloomComplexity}
    To achieve a false positive rate $\epsilon$ when storing $n$ keys in a Bloom Filter, the minimize required bit size $m$ of the Bloom Filter is $m=-\frac{n\ln (\epsilon)}{\ln(2)^2}\approx -1.44n\log_2(\epsilon)$.
\end{lemma}
\begin{proof}
    Let $m$ be the number of bits of the Bloom Filter, and $k$ be the number of independent hash functions employed by Bloom Filter.\\
    Then the probability that a bit is set to 1 by the inserted $n$ keys is $1-(1-\frac{1}{m})^{kn}\approx 1-e^{-\frac{kn}{m}}$.\\
    A false positive occurs if its $k$ mapping bits are all set to 1.
    Therefore, the false positive rate is $(1-e^{-\frac{kn}{m}})^k$.\\
    For given $m$ and $n$, the false positive rate is minimized when $k=\frac{m}{n}\ln 2$, and the corresponding false positive rate is $(\frac{1}{2})^{\frac{m}{n}\ln 2}$.\\
    If we want to achieve a false positive rate $\epsilon$, with the optimal value of $k$, we have $\epsilon = (\frac{1}{2})^{\frac{m}{n}\ln 2}$.\\
    Hence, the least memory is $m=-\frac{n\ln (\epsilon)}{\ln(2)^2}\approx -1.44n\log_2(\epsilon)$ bits, with the corresponding number of hash functions $k=-\frac{\ln(\epsilon)}{\ln (2)}$.
\end{proof}
Lemma \ref{lemma:BloomComplexity} indicates that when keeping less than $n$ distinct item keys, we can add $1.44n$ bits memory and one more hash function to halve the false positive rate.
Since the recovering process of Reversible Bloom Filter is equivalent to the traversal method in the term of result, we need to make the false positive rate be $O(\frac{1}{|\mathcal{K}|})$, where $|\mathcal{K}|$ denotes the cardinality of the key space $\mathcal{K}$.
We therefore prove that the memory cost of Reversible Bloom Filter is bounded by $1.44n\log_2(|\mathcal{K}|)$.
\begin{theorem}
\label{theorem:BitmapBlock}
    Suppose the length of a key is $l$ bits, and we divide the entire key into $k$ segments, each containing $\frac{l}{k}$ bits (here we assume that $\frac{l}{k}$ is an integer).
    We assume that each key segments have independent distribution.
    By using $\frac{n}{W\ln 2}$ blocks, the number of reported keys from the bitmap block array is less than $2^{l-k}$, where $W=2^{\frac{l}{k}}$, and $n$ is the number of distinct keys.
\end{theorem}
\begin{proof}
    Suppose we use $B$ bitmap blocks in the Reversible Bloom Filter.\\
    For the $i$th block, suppose it is mapped by $n_i$ inserting keys.\\
    As the distribution of key segments is not even, we suppose that the probability distribution of a key segment $seg_j$ is $\{p_1,p_2,...,p_{W}\}$, where $p_r$ denotes the probability of the key segment be $r$.\\
    Then the probability of the $r$th bits was set to one in the corresponding bitmap of the block is $1-(1-p_r)^{n_i}$.\\
    Due to the Lyapunov Central Limit Theorem, the number of distinct $seg_j$ is approximately $\sum\limits_{r=1}^{W}(1-(1-p_r)^{n_i})$, with a negligible error probability.\\
    Note that $1-(1-x)^{n_i}$ is a concave function and $\sum\limits_{r=1}^{W}p_r=1$, using Jensen's inequality, we have 
    \begin{equation*}
        \sum_{r=1}^{W}(1-(1-p_r)^{n_i})\leq W(1-(1-2^{-\frac{l}{k}})^{n_i}).\\
    \end{equation*}
    Thus, the number of keys reported by a bitmap block is less than $(W(1-(1-2^{-\frac{l}{k}})^{n_i}))^k\approx ((1-e^{-\frac{n_i}{W}})W)^k$.\\
    By testing whether the key is mapped to its block, and take all the $B$ blocks into consideration, the remaining number of keys is less than $\frac{W^k}{B}\sum\limits_{i=1}^{B}(1-e^{-\frac{n_i}{W}})^k$.\\
    If the hash function is even enough and $n$ is large enough, we have $n_i\approx \frac{n}{B}$, and the total number of keys is less than $W^k(1-e^{-\frac{n}{WB}})^k$.\\
    When $\frac{n}{WB}=\ln 2$, we have the number of keys less than $W^k(1-e^{-\frac{n}{WB}})^k=2^{l-k}$.
\end{proof}
Theorem \ref{theorem:BitmapBlock} has proved that the bitmap block has a comparable memory efficiency to Bloom Filter.
The remaining part of the Reversible Bloom Filter are a set of Bloom Filters depend on different key segments.
We can prove that such a method have negligible accuracy reduction.
\begin{theorem}
    \label{theorem:BloomFilters}
    When using $m$ bits memory, the lowest false positive rate that the Bloom Filters in our method can achieve is approximately $2^{-\frac{m\ln 2}{n}}$.
\end{theorem}
\begin{proof}
    As we have proved in Lemma \ref{lemma:BloomComplexity}, the false positive rate of a Bloom Filter using $m$ bits can achieve a false positive rate of $2^{-\frac{m\ln 2}{n}}$ when storing $n$ items.
    Our Reversible Bloom Filter employs several Bloom Filters depend on key segments.
    To calculate the false positive rate, we consider the Bloom Filters as splitting memory from a entire Bloom Filter based on the entire key to build several child Bloom Filters based on key segments.
    If there are more than two layers of internal nodes, we can do more the splitting operation iteratively, and thus we only discuss one splitting operation here.\\
    Suppose that we allocate $m_1,m_2,...,m_\alpha$ bits for the $\alpha$ child Bloom Filters.
    Let $n_1,n_2,...,n_\alpha$ be the number of distinct value of the key segments that have been inserted to the Bloom Filters.
    Obviously, $n_i\leq n$, the false positive rate of these Bloom Filters can be lower than $2^{-\frac{m_1\ln 2}{n}},2^{-\frac{m_2\ln 2}{n}},...,2^{-\frac{m_\alpha\ln 2}{n}}$.
    Suppose the candidate sets of key segments have $t_1,t_2,...,t_\alpha$ respectively, then we will have the false positive rate less than
    \begin{equation*}
        \frac{(\prod\limits_{i=1}^{\alpha}((t_i-n_i)\times 2^{-\frac{m_i\ln 2}{n}} + n_i) - n)\times 2^{-\frac{m'\ln 2}{n}}}{\prod\limits_{i=1}^{\alpha}t_i-n}
    \end{equation*}
    , where $m'$ is the remaining memory of the parent Bloom Filter.\\
    Slightly modify the expression, we have
    \begin{equation*}
        \frac{(\prod\limits_{i=1}^{\alpha}((\frac{t_i}{n_i}-1)\times 2^{-\frac{m_i\ln 2}{n}} + 1) - \frac{n}{\prod\limits_{i=1}^{\alpha}n_i})\times 2^{-\frac{m'\ln 2}{n}}}{\prod\limits_{i=1}^{\alpha}\frac{t_i}{n_i}-\frac{n}{\prod\limits_{i=1}^{\alpha}n_i}}.
    \end{equation*}
    Note that the candidate set of each key segments is the Cartesian product of the result set of its child segments. Therefore, $t_i$ is a product of several numbers with the same order of $n$. Thus, $\frac{t_i}{n_i}$ is large when $n$ is large. Moreover, $\frac{n}{\prod\limits_{i=1}^{\alpha}n_i}$ is also a negligible term. If we let $(\frac{t_i}{n_i} - 1)\times 2^{-\frac{m_i\ln 2}{n}}$ be much more than 1, we can see that the false positive rate is approximately $2^{-\frac{m\ln 2}{n}}$.
\end{proof}

Combining Theorem \ref{theorem:BitmapBlock} and Theorem \ref{theorem:BloomFilters}, we can deduce the memory overhead of the Reversible Bloom Filter.
\begin{theorem}
    To store and recover $n$ distinct keys with $O(1)$ false positives, the least memory overhead of Reversible Bloom Filter is bounded by $\frac{nl}{\ln 2}\approx 1.44nl$ bits, where $l$ is the number of bits in a key.
\end{theorem}
\begin{proof}
    In Theorem \ref{theorem:BitmapBlock}, we have proved that if we divide the keys into $k$ segments, by using $\frac{kn}{\ln 2}$ bits, we can obtain a candidate key set with less than $2^{l-k}$ keys.
    We can use additional $m=\frac{(k-l)n}{\ln 2}$ bits memory to achieve a $2^{k-l}$ total false positive rate in the remaining Bloom Filters due to Theorem \ref{theorem:BloomFilters}.
    Therefore, to store and recover $n$ distinct keys with $O(1)$ false positives, the least memory overhead is bounded by $\frac{kn}{\ln 2} + \frac{(k-l)n}{\ln 2}=\frac{nl}{\ln2} \approx 1.44nl$ bits.
\end{proof}

We then discuss the number of buckets required by the CM Sketch.
To guarantee a high probability of fully decoding, the CM Sketch should employ enough buckets.
As we have discussed in \S\ref{sec:Alg:Design:DecodingProcess}, if we only use the pure bucket extraction, linear number of buckets is sufficient.
Specifically, when using $k$ hash functions, the probability of failing to decode all the items is bounded by $O(n^{-k+2})$, if $m>c_kn$, where $c_k$ is a constant associated with $k$($c_3=1.222,c_4=1.295,c_5=1.425$).
As we also employ an SVD step and an optional ILP step, the memory overhead is lower.
Although we cannot give the exact bucket number, we conducted a small experiment to show the improvement of these two steps.
We use different numbers of buckets for the CM sketch to record 100 distinct keys.
The frequencies of these keys are generated randomly.
We test the decoding process 1000 times and record the probability of successfully decoding (obtaining the exact frequencies).
As shown in Figure \ref{Fig:successrate}, we find that the two steps significantly improve the success rate.
The ILP step can even decode successfully when the number of buckets is less than the number of keys.
\begin{figure}[t]
\centering
\includegraphics[width=0.4\textwidth]{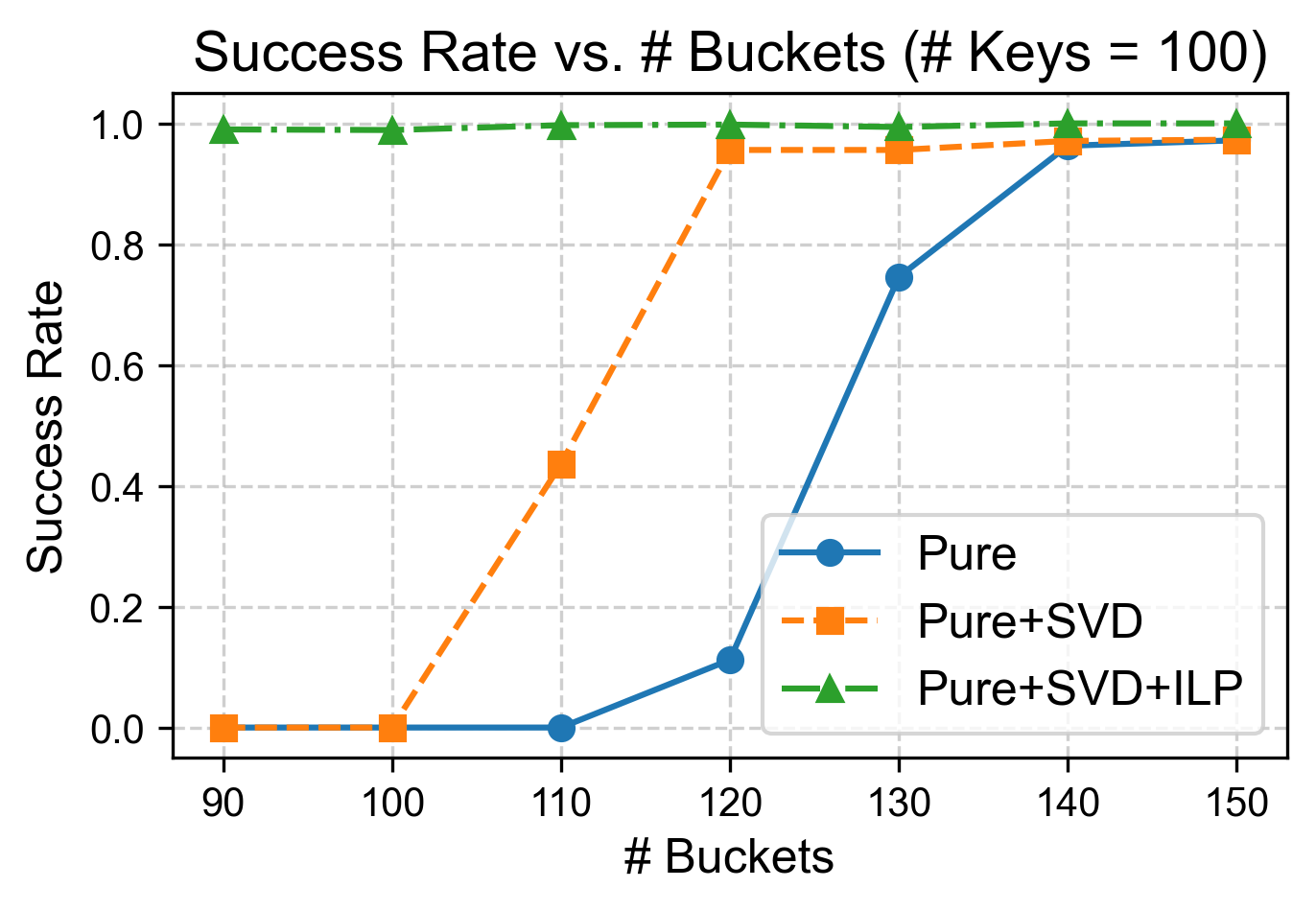} 
\caption{Comparison of success rate.}
\label{Fig:successrate}
\end{figure}

In summary, the total memory requirement for exactly decoding $n$ items of Hidden Sketch is less than $1.23nw+1.44nl$ bits in the best case, where $l$ is the bits of a key, and $w$ is the bits of a bucket.


    